\documentclass[12pt]{article}

\usepackage{epsfig}

\usepackage{amssymb}
\usepackage{amsfonts}

\usepackage{color}
 
\def\be{\begin{equation}}
\def\ee{\end{equation}}
\def\ba{\begin{array}{c}}
\def\ea{\end{array}}

\newcommand{\bea}{\begin{eqnarray}}
\newcommand{\eea}{\end{eqnarray}}

\newcommand{\bbr}{\br\!\br}
\newcommand{\kkt}{\kt\!\kt}
\newcommand{\pbr}{\prec\!\!}
\newcommand{\pkt}{\!\!\succ\,\,}
\newcommand{\kt}{\rangle}
\newcommand{\br}{\langle}

\newtheorem{thm}{Theorem}

\newenvironment{proof}{\noindent
 {\bf Proof.}}{\hfill$\square$\vspace{3mm}\endtrivlist}

\begin{document}

\begin{center}

{\Large \bf

Composite quantum Coriolis forces

}


\vspace{0.8cm}

  {\bf Miloslav Znojil}$^{1,2,3}$

\end{center}


%
%
%
%

\vspace{10mm}

 $^{1}$ {The Czech Academy of Sciences,
 Nuclear Physics Institute,
 Hlavn\'{\i} 130,
250 68 \v{R}e\v{z}, Czech Republic, {e-mail: znojil@ujf.cas.cz}}


 $^{2}$ {Department of Physics, Faculty of
Science, University of Hradec Kr\'{a}lov\'{e}, Rokitansk\'{e}ho 62,
50003 Hradec Kr\'{a}lov\'{e},
 Czech Republic}

$^{3}$ Institute of System Science, Durban University of Technology,
Durban, South Africa



\vspace{10mm}

\section*{Abstract}

In a consistent quantum theory
called ``non-Hermitian interaction picture'' (NIP),
the \textcolor{black}{standard} quantum Coriolis operator $\Sigma(t)$
emerges whenever
the observables of a unitary
system
are given in their
quasi-Hermitian and non-stationary
rather than ``usual''
representations.
\textcolor{black}{With
$\Sigma(t)$ needed, in NIP, in both
the Schr\"{o}dinger-like and
Heisenberg-like dynamical evolution equations}
we show that another, amended and potentially simplified theory
can be based on an auxiliary
$N-$term factorization of the
Dyson's Hermitization map $\Omega(t)$. The knowledge of this
\textcolor{black}{factorization is shown to
lead to a multiplet of alternative eligible
Coriolis forces $\Sigma_{n}(t)$ with $n=0,1,\ldots,N$}.
The related formulae for the measurable
predictions \textcolor{black}{constitute a new
formalism called ``factorization-based non-Hermitian
interaction picture'' (FNIP)}.
The \textcolor{black}{conventional NIP
formalism (where $N=1$)} becomes complemented \textcolor{black}{
by an $(N-1)$-plet of its innovative ``hybrid''}
alternatives. Some of the respective
{\it ad hoc\,} adaptations of observables
may result in
\textcolor{black}{an optimal representation of quantum dynamics}.

\newpage

\section{Introduction}

\textcolor{black}{In this paper a
thoroughly innovated formulation of quantum mechanics
(to be called
``factorization-based non-Hermitian interaction picture'', FNIP)
will be proposed and described.
In contrast to the majority of the older textbooks
offering a conventional
presentation of
quantum theory
(cf., e.g., \cite{Messiah}),
the operators of observables will be assumed non-Hermitian
and non-stationary.
In this context, reviews \cite{Geyer,Carl,SIGMA,ali} and \cite{book}
or, at least,
sections~\ref{druhasekce} and~\ref{sectiontwo}
below may be recommended for introductory reading.
}

\textcolor{black}{The dominant inspiration of our present
innovative methodical proposal
may be seen in the existence of certain exactly solvable unitary
quantum models in which the evolution is controlled by a stationary
non-Hermitian Hamiltonian $H \neq H(t)$ (for an example see
Appendix A below) and, more importantly, in the existence of
the exactly solvable unitary
quantum models in which the dynamics is represented by a set of
certain
manifestly time-dependent, non-stationary
non-Hermitian operators of observables denoted, say,
as $\Lambda_0(t)$, $\Lambda_1(t)$ etc (see, e.g., review~\cite{SIGMA}
for more details).
}

\textcolor{black}{Our interest in the
latter and, admittedly, important but still
slightly exotic subclass of the solvable non-Hermitian quantum models
was motivated, first of all, by the related mathematics.
Indeed, in both the stationary and non-stationary scenarios one must,
after all, re-Hermitize the
observables in terms of a respective stationary
Dyson map $\Omega$ \cite{Geyer,ali,Dyson}
or
of its time-dependent analogue
$\Omega(t)$ \cite{SIGMA,timedep,FringMou}.
Such a Hermitization
has the form
 \be
 \Lambda_k(t)=\Omega(t)\,\lambda_k(t)
 \Omega^{-1}(t)\,\ \ \ {\rm with} \ \ \ \lambda_k(t)=\lambda_k^\dagger(t)\,,
 \ \ \ \ \ k=0,1,\ldots
 \,
 \label{usuri}
 \ee
of a suitable explicit (or, at least, possible) transformation
of all of the operators of observables.
For consistency reasons it is unavoidable \cite{Geyer} but in
applications it represents
one of the decisive technical challenges.}

\textcolor{black}{
In such a context
our attention has been attracted by a more or less purely empirical
observation that
in virtually all of the above-mentioned solvable models the authors managed to
circumvent the technical obstacles via an
auxiliary Dyson-map-factorization assumption
 \be
 \Omega(t)=\Omega_N(t)\,\Omega_{N-1}(t)\,\ldots \,\Omega_{1}(t)\,.
 \label{fffs}
 \ee
A deeper study of this coincidence has been performed in \cite{[380],[381]}.
For the sake of simplicity we decided to keep the Dyson map stationary,
$
\Omega(t)=\Omega(0)=\Omega
$, i.e., we decided to work just in
the simpler
formulation of quantum mechanics called
``non-Hermitian Schr\"{o}dinger picture'' (NSP).
The decision proved productive
because we revealed that
after the
factorization (\ref{fffs}) at a nontrivial $N \geq 2$,
the two conventional and well known
forms of the NSP theory (cf. \cite{ali})
become complemented
by a fairly rich menu of its $N-1$
alternatives
called ``hybrid NSP'' (HNSP).
Via a few schematic toy-model examples we also demonstrated,
in \cite{[380],[381]},
that the availability of such a menu
opens, for a given stationary input Hamiltonian $H$,
the possibility of the choice
of an ``optimal'' version of the HNSP-based
strategy of the quantum model-building.
}

\textcolor{black}{
The stationarity constraints $H \neq H(t)$ and $\Omega \neq \Omega(t)$
have subsequently been removed
in \cite{[382]}. We choose there $N=2$, and we
managed to find, describe and illustrate there
the specific merits of a unique ``hybrid'' $N=2$
reformulation of quantum mechanics
in non-Hermitian interaction picture.
Paradoxically, one of the weaknesses of the latter
result lied in
the process of its derivation
because it seemed to indicate that all of
its $N>2$
generalizations would be trivial.
Later on we imagined that this is not the case.
On the contrary, we found that the transition
from the choice of $N=2$
to the general $N>2$ scenarios
proved truly remarkable
from both the mathematical and phenomenological points of view.
And this is also the main
message to be delivered by
our present paper.
}

The presentation of our results will
be preceded by the introductory
review of the state of art in
sections \ref{druhasekce}
and \ref{sectiontwo}.
The first part of the new, $N>2$ amendment of the
non-stationary, NIP theory will be then, in
section \ref{sectionthree},
devoted to the construction
of the general composite Coriolis forces and to the
explanation of their role
in the Schr\"{o}dinger-like NIP evolution equations.

In
section \ref{sectionfour}
this result will be complemented by
the derivation of the $N>2$ form of the
Heisenberg-like equations controlling the
characteristic NIP
time-dependence of the
operators representing the observables.
Moreover,
an outline of dynamics in
the statistical quantum mechanics will be provided,
with the state of the system being merely
specified as
a mixture of pure states (i.e., by the so called
density matrix).

\textcolor{black}{
In general, the proposed FNIP formalism is supposed to be useful,
efficient and applicable to
any unitary quantum system in which there appears
the possibility of a simplification of its dynamical evolution
equations. This is to be realized
by means of a suitable factorization (\ref{fffs}) of the
overall non-stationary Dyson map $\Omega(t)$
entering the simplification ansatz (cf. Eq.~(\ref{anz}) below).
An
explicit illustrative example of such an eligible
system will be provided and described in section \ref{appb}.
We will show there that for the specific non-stationary
non-Hermitian
anharmonic-oscillator model of Fring and Tenney \cite{FT},
the fundamental factorization (\ref{fffs})
can even be realized
non-numerically and
in several alternative arrangements.
To such a climax of our paper an}
extensive discussion is added in
section~\ref{sectionfive}, and our message is finally summarized in
section~\ref{sectionsix}.

\section{\textcolor{black}{Classical Coriolis force and its quantum analogue.}\label{druhasekce}}

The phenomenon called
``Coriolis force'' \cite{Coriolis,Coriolisb} emerges,
in classical mechanics, due to the
non-inertial choice
of the frame of reference.
A quantum analogue of such a force
(i.e., the main subject of our present interest)
enters the game when one considers
Schr\"{o}dinger equation
 \be
 {\rm i}\frac{d}{dt}\,|\psi(t)\pkt=
  {}{\mathfrak{h}}(t)\,|\psi(t)\pkt
 \label{sekv}
 \ee
in
a conventional Hilbert space of states
(say, ${\cal L}$)
in which the Hamiltonian is assumed self-adjoint,
  \be
  {}{\mathfrak{h}}(t)=
  {}{\mathfrak{h}}^\dagger(t)\,.
  \label{hershaft}
  \ee
Such a constraint is equivalent to the unitarity of
the evolution \cite{Stone} so that, in some sense, it
specifies a quantum analogue of
the classical choice of the inertial reference frame.

A quantum analogue of transition to the
non-inertial frame
\textcolor{black}{may be interpreted as mimicking a
gravitational force (see a deeper development of this idea in
\cite{ju22}).
In the language of mathematics the parallelism}
can be encountered when one preconditions the
ket-vector element of ${\cal L}$ using the Dyson-inspired map,
 \be
 |\psi(t)\pkt = \Omega(t)\,|\psi(t)\kt\,.
 \label{anz}
 \ee
One may expect the emergence of
a quantum analogue of the Coriolis force
in the various
realistic applications of preconditioning (\ref{anz})
\textcolor{black}{(see, e.g., \cite{[thi5]})}.
Indeed, the
invertible non-stationary operator $\Omega(t)$
may simulate important features of dynamics, provided only that
the
transformations $\Omega(t)$ are
sufficiently flexible.

In
a way recommended and tested by Dyson \cite{Dyson}
such operators may even be
non-unitary yielding nontrivial product
 \be
 \Omega^\dagger(t)\,\Omega(t)=\Theta(t)\,.
 \label{nonuni}
 \ee
Even when $\Theta(t) \neq I$
(i.e., even when the preconditioning is non-unitary),
the symbol $|\psi(t)\kt$ may still be treated
as a useful representative of the state
of the quantum system \cite{Geyer}.
In general, it can be interpreted as a
ket-vector element
of an auxiliary, ``fictious'', manifestly unphysical
but mathematically friendlier
Hilbert space (say, ${\cal F}$) \cite{SIGMA}.
Subsequently, the insertion of ansatz (\ref{anz})
converts the textbook Schr\"{o}dinger
Eq.~(\ref{sekv}) living in ${\cal L}$ into its equivalent
version in ${\cal F}$,
 \be
 {\rm i}\frac{d}{dt}\,|\psi(t)\kt= G(t)\,|\psi(t)\kt\,,
 \ \ \ \
 G(t)=H(t)-\Sigma(t)\,,
 \ \ \ \ H(t)=\Omega^{-1}(t)\,
 {}{\mathfrak{h}}(t)\,\Omega(t)\,
 \label{resekv}
 \ee
where
 \be
 \Sigma(t)=\frac{\rm i}{\Omega(t)}\,\dot{\Omega}(t)\,,\ \ \ \
 \dot{\Omega}(t)=
 \frac{d}{dt}\Omega(t).
 \label{doresekv}
 \ee
The latter operator $\Sigma(t)$ defined
in terms of the
known non-stationary Dyson map $\Omega(t)$
is precisely the
quantum analogue of the classical Coriolis force.

Equation (\ref{resekv}) plays the role of
a starting point of
a
reformulation of quantum mechanics
in non-Hermitian interaction picture (NIP)
as described
in paper \cite{timedep}
and reviews \cite{book,FringMou,NIP,Bishop,Andreas}.
In comparison,
the older, stationary
non-Hermitian Schr\"{o}dinger picture (NSP,
cf. reviews \cite{Geyer,Carl,ali}) is much simpler to apply.
In particular,
in NSP one finds that the
Coriolis forces are absent, $\Sigma(t)=0$.
As a consequence,
the NSP generator $G(t)$ in (\ref{resekv}) coincides then with
the hiddenly Hermitian NSP Hamiltonian $H(t)$.

Certainly, the non-stationary approach
is technically much more
challenging.
One is forced to
search,
more intensively,
for the
simplifications of the NIP formalism.
In this
sense,
the main purpose of our present paper
can briefly be formulated as a transfer of some of the most recent
technical amendments
of the stationary NSP quantum theory
to the non-stationary
NIP context.

In the early stages of such a project,
naturally, we were well aware
of the fact that in the non-stationary NIP framework
the emergence of the Coriolis forces would really lead to a number of the
new technical challenges.
This was one of the reasons why we did not dare to move, in
the subsequent non-stationary framework of Ref.~\cite{[382]},
beyond the first non-trivial non-stationary $N=2$ setup.
Nevertheless,
the appeal and urgency of some of the related realistic applications
of the NIP approach
(say, in the coupled-cluster context of review \cite{Bishop})
forced us to return to the description of the extension of the
NIP analysis of paper \cite{[382]}
to the general case in which one would be able to use the more general
ansatz (\ref{fffs}) with $N>2$.

\section{Composite Dyson maps (stationary case)\label{sectiontwo}}

Although the abstract
formulation of the
stationary version of the
quasi-Hermitian quantum mechanics (QHQM, \cite{Geyer})
is now more than thirty years old,
the widely accepted opinion is that ``it is
generally very difficult to implement'' (cf. p. 1216 in \cite{ali}).
Still, after a number of technical amendments
\cite{Carl,ali,book,Carlbook} the related
model-building strategy became enormously popular \cite{book}.
One of the reasons may be
seen in the extremely helpful assumption
of
the stationarity of the Dyson
map
$\Omega(t)=\Omega(0)=\Omega$
as well as of the related
operator $\Theta(t)=\Theta(0)=\Theta\,$
called the Hilbert-space metric.
In some models, indeed,
this made the
practical applicability of the
QHQM approach
(mentioned in the second line of Table~\ref{hybr})
comparable with the use of the
conventional
quantum theory using Hermitian Hamiltonians
and mentioned in the first line of Table~\ref{hybr}.

\begin{table}[h]
\caption{Three alternative realizations of unitarity in stationary NSP framework.
}
 \label{hybr}
\vspace{0.3cm}
\centering
\begin{tabular}{||l|l|l|l||}
\hline \hline
 \multicolumn{1}{||c|}{quantum theory}
 & \multicolumn{1}{|c|}{Hamiltonian}
  &  \multicolumn{1}{|c|}{metric}
  &
    \multicolumn{1}{|c||}{condition}
 \\
 \hline \hline
 conventional, \cite{Messiah}
   &
 $\mathfrak{h}$&
 $I$
 &
 $\mathfrak{h}^\dagger=\mathfrak{h}$
  \\
  Dyson-inspired, \cite{Geyer}
 &$H=\Omega^{-1}\,\mathfrak{h}\,\Omega$&
 $\Theta=\Omega^\dagger\,\Omega$
 &
 $H^\dagger\,\Theta=\Theta\,H$
 \\
 hybrid  \cite{[380]},
  $\,\Omega=\Omega_2\,\Omega_1$
 &
 $H_{1}=\Omega_1\,H\,\Omega_1^{-1}$&
 $\Theta_2=\Omega_2^\dagger\,\Omega_2$
 &
 $H_{1}^\dagger\,\Theta_2=\Theta_2\,H_1$
 \\
 \hline \hline
\end{tabular}
\end{table}

In the stationary QHQM {\it alias\,} NSP framework
the unitary evolution of the quantum system in question
is in fact represented, simultaneously, in the
two non-equivalent Hilbert spaces,
viz., in the unphysical
space
${\cal F}={\cal H}_{mathematical}$
(with the conventional trivial
identity-operator metric) and in its
``physical-interpretation'' partner (say, ${\cal H}$
or ${\cal H}_{physical}$,
differing from ${\cal F}$
just by an amendment of the inner product, $I \to \Theta$).
From a purely pragmatic perspective such an enhancement of the
flexibility of the formalism
(and, presumably, a significant simplification of the Hamiltonian)
just transferred the
potential mathematical difficulties (if any)
from the conventional self-adjoint operator $\mathfrak{h}$
to the new nontrivial self-adjoint operator $\Theta$.

In a brief letter~\cite{[380]}
we proposed, therefore,
another, ``hybrid'' theoretical option
(cf. the third line of Table~\ref{hybr}).
A
more detailed
exposition
of the idea was then provided in
\cite{[381]}. We obtained
a reformulation of the
unitary theory
called hybrid non-Hermitian Schr\"{o}dinger picture (HNSP).

A formal mathematical guide to the latter amendment of the theory
has been found in the assumption
of an $N-$term factorization
of the
Dyson map, i.e., in a time-independent
special case of Eq.~(\ref{fffs}),
 \be
 \textcolor{black}{\Omega=\Omega_N\,\Omega_{N-1}\,\ldots \,\Omega_{1}\,.}
 \label{sfffs}
 \ee
In Ref.~\cite{[380]} we outlined only
the main consequences of
the special, $N=2$ realization of such a
factorization.
A more consequent implementation of the idea
has been then performed in \cite{[381]}.
An
extension
of applicability of
the $N=2$ ansatz $\Omega=\Omega_2\,\Omega_{1}$ of Ref.~\cite{[380]}
to any integer $N$
has been achieved there in the form which is
recalled here
in Table \ref{enhybr}.
The parameter $j=0,1,\ldots,N$ numbers there the
alternative formulations of the
unitary quantum theory in its
mathematical (i.e., quasi-Hermitian) representation in
${\cal F}$.

\begin{table}[h]
\caption{Eligible versions of the \textcolor{black}{hybrid}
NSP Hamiltonians.  At
any $N$ and $j$,
the  \textcolor{black}{related hybrid} inner-product metrics
are defined in terms of
products
$\,\Omega_{[j]}=\Omega_N\,\Omega_{N-1}\,\ldots
 \,\Omega_{j+1}$
(cf.~\cite{[381]}).
}
 \label{enhybr}
\vspace{0.3cm}
\centering
\begin{tabular}{||c||l|l|l||}
\hline \hline
  $j$
 & \multicolumn{1}{|c|}{Hamiltonian} &
    \multicolumn{1}{|c|}{physical metric}
    &
    \multicolumn{1}{|c||}{ quasi-Hermiticity }
 \\
\hline \hline
 0
 &$H=H_0 $ (non-hybrid case)
 &$ \Theta=\Theta^{[0]}=\Omega_{[0]}^\dagger\,\Omega_{[0]} $
 &
 $H^\dagger\,\Theta=\Theta\,H$
 \\
 \hline
 1
 &$H_1=\Omega_{1}\,H_0\,\Omega_{1}^{-1}
 $
 &$\Theta^{[1]}=\Omega_{[1]}^\dagger\,\Omega_{[1]}$
 &
 $H^\dagger_1\,\Theta^{[1]}=\Theta^{[1]}\,H_1$
 \\
 2
 &$H_2=\Omega_{2}\,H_1\,\Omega_{2}^{-1}
 $
 &$\Theta^{[2]}=\Omega_{[2]}^\dagger\,\Omega_{[2]}$
 &
 $H^\dagger_2\,\Theta^{[2]}=\Theta^{[2]}\,H_2$
 \\
 ...&...&...&...\\
 $\!N$-$2\!$
 &$H_{N-2}=\Omega_{N-2}\,H_{N-3}\,\Omega_{N-2}^{-1}
 $
 &$\Theta^{[N-2]}=\Omega_{[N-2]}^\dagger\,\Omega_{[N-2]}$
 &
 $H^\dagger_{N-2}\,\Theta^{[N-2]}=\Theta^{[N-2]}\,H_{N-2}$
 \\
 $\!N$-$1\!$
 &$H_{N-1}=\Omega_{N-1}\,H_{N-2}\,\Omega_{N-1}^{-1}
 $
 &$\Theta^{[N-1]}=\Omega_{N}^\dagger\,\Omega_{N}$
 &
 $H^\dagger_{N-1}\,\Theta^{[N-1]}=\Theta^{[N-1]}\,H_{N-1}$
 \\
 \hline
 $N$
 &$\mathfrak{h}=\Omega_{N}\,H_{N-1}\,\Omega_{N}^{-1}\equiv H_N
 $& $I$ (Hermitian theory)
 &
 n.a.,  $\mathfrak{h}=\mathfrak{h}^\dagger$\\
 \hline \hline
\end{tabular}
\end{table}

In the Table
the basic features of the resulting general
hybridized QHQM theory
are given a recurrent form
which complements the recipe of Ref.~\cite{[381]}.
In the language of mathematics
just a more systematic partial transfer of the information about dynamics
from the Hamiltonian to the metric is performed.
This means that
every choice of $j$
specifies a different
physics-representing
Hilbert space
${\cal R}_{0}^{[j]}$.

Such double-indexed symbols denoting the individual
Hilbert spaces of interest were initially introduced, with a more abstract
motivation, in
\cite{[381]}.
Our present $(N+1)-$plet of the {\em physical\,} Hilbert
spaces can be found there to be just a subset of
a larger multiplet
arranged in a triangular-shaped lattice (cf. the
diagram~Nr.~32 in~\cite{[381]}).
Incidentally, all of the
other Hilbert-space elements
${\cal R}_{k}^{[j]}$ of the lattice with $k>0$
appeared
manifestly unphysical.

For the state-vector  elements of
our present physical $k=0$ Hilbert spaces
it appeared necessary to use
a specific, amended Dirac's bra-ket notation
convention. Its basic features are summarized in Table~\ref{huhybr}.
In particular, using an additional biorthonormality assumption,
also
a formal spectral representation of the Hamiltonians
becomes available.

\begin{table}[h]
\caption{Kets and bras of the physical Hilbert spaces ${\cal R}_0^{[j]}$
as represented in ${\cal F}$.
 }
 \label{huhybr}
\vspace{0.3cm}
\centering
\begin{tabular}{||c|c|c|c||}
\hline \hline
    \multicolumn{1}{||c|}{$j$}
 & \multicolumn{1}{|c|}{ket
 }& \multicolumn{1}{|c|}{bra
 }
  &  \multicolumn{1}{|c||}{ Hamiltonian
  }
 \\
\hline \hline
 0
 &$|_{[0]}\psi\kt \equiv |\psi\kt
 $
 &$\bbr_{[0]} \psi|=\br_{[0]}\psi|\,\Theta^{[0]}\!\equiv\!\bbr \psi|$
 &
 $H_0=\sum_\centerdot \,|_{[0]}\centerdot\kt E_\centerdot \bbr_{[0]} \centerdot| \equiv H$
 \\
 \hline
 1&$|_{[1]}\psi\kt=\Omega_1\,|_{[0]}\psi\kt$
 &$\bbr_{[1]} \psi|=\br_{[1]}\psi|\,\Theta^{[1]}$
 &
 $H_1=\sum_\centerdot \,|_{[1]}\centerdot\kt E_\centerdot \bbr_{[1]} \centerdot|$
 \\
 2
 &$|_{[2]}\psi\kt=\Omega_2\,|_{[1]}\psi\kt$
 &$\bbr_{[2]} \psi|=\br_{[2]}\psi|\,\Theta^{[2]}$
 &
 $H_2=\sum_\centerdot \,|_{[2]}\centerdot\kt E_\centerdot \bbr_{[2]} \centerdot| $
 \\
 ...&...&...&...\\
 $\!N\!-\!1\!$
 &$|_{[N-1]}\psi\kt=\Omega_{N-1}\,|_{[N-2]}\psi\kt$
 &$\bbr_{[N-1]} \psi|=\br_{[N\!-\!1]}\psi|\,\Theta^{[N\!-\!1]}$
 &
 $H_{N\!-\!1}=\sum_\centerdot \,|_{[N\!-\!1]}\centerdot \kt E_\centerdot \bbr_{[N\!-\!1]} \centerdot| $
 \\
 \hline
 $N$
 &$|_{[N]}\psi\kt\!=\Omega_N|_{[N\!-\!1]}\psi\kt\!\equiv\!|\psi\pkt\!$
 &$\!\bbr_{[N]} \psi|\!=\!\br_{[N]}\psi|\!\equiv\,\pbr \psi|\!$
 &
 $H_N=\sum_\centerdot \,|_{[N]}\centerdot\kt E_\centerdot \bbr_{[N]}
  \centerdot| \equiv \mathfrak{h}$
 \\
 \hline \hline
\end{tabular}
\end{table}

In spite of the full phenomenological equivalence of
the formulations of the theory at the different choices of $j$
(i.e., in spite of the $j-$independence
of the predictions of the theory)
the feasibility of the
practical mathematical search for the solutions of
Schr\"{o}dinger
equations living in the respective $j-$superscripted
Hilbert space ${\cal R}_{0}^{[j]}$
may vary quite strongly with $j$.
A
persuasive illustration of the phenomenon
was provided, via
several toy-model examples, in \cite{[380],[381],[382]}.
In most cases, unfortunately,
an optimal representation index $j$ may only be specified {\it a posteriori}.

In our present methodical considerations
it makes sense to
consider, therefore, both $N$ and $j$
as
quantities, an optimal choice of which crucially depends
on a preselected Hamiltonian.
In the stationary NSP
interpretation of all of
the three QHQM Tables \ref{hybr}, \ref{enhybr} and \ref{huhybr}
the individual operators of
the Hilbert-space metric were also assumed time-independent.
Nevertheless, the same Hilbert-space-classification pattern can be
equally well used
in the more general, non-stationary quantum theory
of our present interest.

\section{\textcolor{black}{Non-stationary systems and composite Coriolis forces}\label{sectionthree}}

In the non-stationary theory
the first feature to notice is
the emergence of the quantum Coriolis forces.
As a consequence, in a way paralleling the derivation of Eq.~(\ref{resekv}),
once we admit the manifest time-dependence of the
separate factor operators $\Omega_n=\Omega_n(t)$
and of their products $\Theta^{[j]}=\Theta^{[j]}(t)$,
the related Hamiltonians $H_j(t)$ themselves lose their
role of the generators of the evolution \textcolor{black}{of the states
in ${\cal R}_{0}^{[j]}$ at $j=0,1,\ldots,N$}.

\subsection{Schr\"{o}dinger equations for vectors $|_{[j]}\psi(t)\kt$}

The basic details of
the necessary upgrade of the theory
may be explained even without reference to the
factorization ansatz~(\ref{fffs}).
We may set $N=1$
and re-derive immediately the time-dependent
Schr\"{o}dinger Eq.~(\ref{resekv})
containing not only the general NIP Hamiltonian $H(t)$ but also
the general NIP Coriolis force of Eq.~(\ref{doresekv}).
Naturally. such a
version of the
non-stationary NIP formalism
may already be found described in the literature
(cf., e.g., \cite{NIP}).
Now, our attention has to be redirected to
the non-stationary
factorization-related theories which are
characterized by the choice of
a suitable larger $N \geq 2$.
Due to such a decision,
the two existing and sufficiently well known QHQM
model-building strategies
(corresponding to the choices of
$j=0$ and $j=N$ in Table \ref{enhybr})
have to be complemented by all of the
remaining eligible
options with $0< j< N$.

After a non-stationary
reinterpretation of symbols in Tables \ref{enhybr} and \ref{huhybr}
let us now fix the value of $N$
and describe the consequences. In a preparatory step
we only have to choose $j=N$ and
rewrite the conventional Schr\"{o}dinger Eq.~(\ref{sekv})
using an innovated, more explicit notation,
 \be
 {\rm i}\frac{d}{dt}\,|_{[N]}\psi(t)\kt=
  H_N(t)\,|_{[N]}\psi(t)\kt\,,\ \ \ \
  |_{[N]}\psi(t)\kt\ \in \ {\cal R}_0^{[N]}
  \,.
 \label{asekv}
 \ee
The operator $H_j(t)$ with maximal $j=N$ is
exceptional by being self-adjoint
not only
in the $N-$th physical Hilbert space ${\cal R}_0^{[N]}$
but also
in our ``mathematical'', auxiliary representation space ${\cal F}$.
The ``trivial'' conventional quantum mechanics of textbooks is obtained.
The knowledge of the Hamiltonian and of the
initial value of $|_{[N]}\psi(t)\kt$ at $t=0$ enables us to reconstruct the state
vector  $|_{[N]}\psi(t)\kt$ at any subsequent time $t>0$.

Once we
turn attention to the first nontrivial
reformulation of the NIP theory in the
next physical Hilbert space ${\cal R}_0^{[j]}$ with $j=N-1$, we may
recall relation
  \be
   H_N^\dagger(t)=
  H_N(t)=\Omega_N(t)\,H_{N-1}(t)\,\Omega_N^{-1}(t)\,.
  \label{uhersha}
  \ee
Having, in parallel, formula $|_{[N]}\psi(t)\kt=\Omega_N(t)\,|_{[N-1]}\psi(t)\kt$
and introducing the following family of the auxiliary tilded operators
 \be
 \widetilde{\Sigma}_{j}(t)
 =
 \frac{\rm i}{\Omega_{j}(t)}\,\dot{\Omega}_{j}(t)\,,\ \ \ \
 \dot{\Omega}_{j}(t)=
 \frac{d}{dt}\Omega_{j}(t)\,,\ \ \ \ j = N, N-1, \ldots, 2, 1\,,
 \label{udoresekv}
 \ee
we may
define the initial Coriolis force
${\Sigma}_{N}(t)=\widetilde{\Sigma}_{N}(t)$ and
rewrite the conventional Hermitian Schr\"{o}dinger Eq.~(\ref{asekv})
of textbooks as another equation living in
the neighboring, modified Hilbert space ${\cal R}_0^{[N-1]}$
with ket-vector elements $|_{[N-1]}\psi(t)\kt$,
 \be
 {\rm i}\frac{d}{dt}\,|_{[N-1]}\psi(t)\kt=
  G_{N-1}(t)\,|_{[N-1]}\psi(t)\kt\,,\ \ \ \
    \,
 G_{N-1}(t)=H_{N-1}(t)-\Sigma_N(t)\,.
 \label{uasekv}
 \ee
In the special case of $N=1$ the new generator $G_{N-1}(t)$
of the evolution of the kets
in ${\cal R}_0^{[N-1]}$
coincides with operator $G(t)$ of Eq.~(\ref{resekv}) of course.

Whenever $N>1$, the $j=N-1$ version~(\ref{uasekv})
of the NIP Schr\"{o}dinger equation is new.
It
differs not only
from
its conventional $j=N$ predecessor~(\ref{asekv})
but also from
its ultimate QHQM $j=0$ alternative (\ref{resekv}).
In a way paralleling
the stationary case,
the corresponding version of quantum theory
may be assigned the acronym ``hybrid NIP'' {\it alias\,} ``HIP'', therefore.

It is worth adding that
in a way illustrated by an elementary toy model in \cite{[382]},
the spectrum of
the HIP generators $G_{j}(t)$ with $j<N$ may be complex
irrespectively of the  possible unitarity of
the evolution of the physical quantum system
in question.
In other words,
even in the unitary dynamical regime
the ``false non-stationary
Hamiltonian'' operators $G_{j}(t)$ with $j<N$ cannot be
interpreted as quasi-Hermitian.

Our comment
was formulated in the $j=N-1$
context but it
applies
to the HIP theory
at any smaller~$j$.
In particular, in order to move
from $j=N-1$
to $j=N-2$ at $N>1$
we may set
$|_{[N-1]}\psi(t)\kt=\Omega_{N-1}(t)\,|_{[N-2]}\psi(t)\kt$.
This enables us to convert quasi-Hermitian Eq.~(\ref{uasekv})
into a new, modified form of Schr\"{o}dinger equation living in
the next physical Hilbert space,
 \be
 {\rm i}\frac{d}{dt}\,|_{[N-2]}\psi(t)\kt=
  G_{N-2}(t)\,|_{[N-2]}\psi(t)\kt\,,\ \ \ \
 G_{N-2}(t)=H_{N-2}(t)-{\Sigma}_{N-1}(t)\,
 \label{vuasekv}
 \ee
where $
  |_{[N-2]}\psi(t)\kt\ \in \ {\cal R}_0^{[N-2]}
  \,$ and where
 \be
 {\Sigma}_{N-1}(t)=\widetilde{\Sigma}_{N-1}(t)+
 \Omega_{N-1}^{-1}(t)\,\Sigma_{N}(t)\,\Omega_{N-1}(t)\,.
 \label{tvudoresekv}
 \ee
This result indicates that and how the construction can be continued iteratively.

\begin{thm}
\label{thmjedna}
In the $j-$th physical Hilbert space ${\cal R}_0^{[j]}$ with $j=N-k$,
the Schr\"{o}dinger equation controlling the time-evolution of the state-vector kets
$|_{[N-k]}\psi(t)\kt\ \in \ {\cal R}_0^{[N-k]}$ has the form
 \be
 {\rm i}\frac{d}{dt}\,|_{[N-k]}\psi(t)\kt=
  G_{N-k}(t)\,|_{[N-k]}\psi(t)\kt\,,\ \ \ \
  \,
  G_{N-k}(t)=H_{N-k}(t)-{\Sigma}_{N-k+1}(t)\,\label{gvuasekv}
 \ee
where formula
 \be
 {\Sigma}_{N-k+1}(t)=\widetilde{\Sigma}_{N-k+1}(t)+
 \Omega_{N-k+1}^{-1}(t)\,{\Sigma}_{N-k+2}(t)\,\Omega_{N-k+1}(t)\,
 \label{ktvudoresekv}
 \ee
defines the necessary composite Coriolis force recursively.
\end{thm}
\begin{proof}
\textcolor{black}{The proof is elementary, provided by the text preceding the theorem.}
\end{proof}

\subsection{Schr\"{o}dinger equations for vectors $|_{[j]}\psi(t)\kkt$}

Let us now turn attention to the formulae listed in the
bra-vector column of Table \ref{huhybr}.
For the sake of brevity we will not add and write
the time-dependence arguments $(t)$ explicitly,
having still in mind the generic
non-stationarity of the underlying NIP and HIP formulations
of the theory.

As long as
all of our mathematical manipulations are assumed represented
in the unique and time-independent auxiliary (i.e., unphysical
but user-friendly)
Hilbert space ${\cal F}$,
we may amply use the fact that in such a space
the inner-product metric is trivial. This means that
the Hermitian conjugation converts
every state-vector bra $\bbr_{[j]}\psi(t) | $
into its dual, twice-marked
ket-vector representation $|_{[j]} \psi(t) \kkt$,
slightly easier to deal with.
Thus, at $j=N$, $N-1$, \ldots we have
 \be
 |_{[N]}\psi\kkt =  |_{[N]}\psi\kt\ =|\psi\pkt\,,
 \ \ \
 |_{[N-1]}\psi\kkt = \Omega_N^\dagger\, |_{[N]}\psi\kkt\,,
 \ \ \
 |_{[N-2]}\psi\kkt = \Omega_{N-1}^\dagger\, |_{[N-1]}\psi\kkt\,
 \ee
etc. The elementary form of these recurrences enables us to
prove not only that Eq.~(\ref{asekv})
may be assigned its trivial conjugate partner
 \be
 {\rm i}\frac{d}{dt}\,|_{[N]}\psi(t)\kkt=
  H_N^\dagger(t)\,|_{[N]}\psi(t)\kkt\,,\ \ \ \
  |_{[N]}\psi(t)\kkt\ \in \ {\cal R}_0^{[N]}
  \,
 \label{pasekv}
 \ee
but also that such a partnership extends, in recurrent manner,
to all of the
other relevant physical Hilbert spaces.

\begin{thm}
\label{thmdve}
In the $j-$th physical Hilbert space ${\cal R}_0^{[j]}$ with $j=N-k$,
the time-evolution of the
twice-marked state-vectors $|_{[N-k]}\psi(t)\kkt\ \in \ {\cal R}_0^{[N-k]}$
is controlled by
Schr\"{o}dinger equation
 \be
 {\rm i}\frac{d}{dt}\,|_{[N-k]}\psi(t)\kkt=
  G_{N-k}^\dagger(t)\,|_{[N-k]}\psi(t)\kkt\,
  ,\ \ \ \
  k=0,1,\ldots,N\,.
 \label{sgvuasekv}
 \ee
\end{thm}
\begin{proof}
One of the easiest ways of the construction of
the set of equations (\ref{sgvuasekv})
may start from the relation
 \be
 |_{[N-k]}\psi(t)\kkt=\Omega^\dagger_{N-k+1}(t)\,|_{[N-k+1]}\psi(t)\kkt
 \ee
which follows immediately from the definitions listed in the third column
of Table~\ref{huhybr}.
\end{proof}

\section{Operator evolution equations\label{sectionfour}}

In preceding section
we explained
that
the stationary NSP formulae of
Table \ref{enhybr}
may be assigned their very natural
non-stationary NIP
analogues.
Now we have to add that at every index $j$
of the physical Hilbert space ${\cal R}_0^{[j]}(t)$, the assumption of
the time-dependence of the related physical inner-product metric $\Theta^{[j]}(t)$
is connected with the generic
time-dependence of the energy-observable $H_j(t)$
(i.e., in our present terminology, of the Hamiltonian).
Simultaneously, it is necessary to imagine that
the
time-dependence of the same metric
must be also compatible with the possible
time-variability
of any other admissible observable, say,
$A_j(t)$.

\subsection{Heisenberg equations.\label{parHei}}

Necessarily \cite{Geyer,ali}, the latter operator must have
real spectrum and it must be
quasi-Hermitian,
 \be
 A_j^\dagger(t)\,\Theta^{[j]}(t)=\Theta^{[j]}(t)\,A_j(t)\,.
 \ee
By construction, the standard probabilistic physical
interpretation of the theory
is in fact in a one-to-one correspondence with the
validity of the latter
postulate.
Naturally, the non-stationarity of the metric
leads to certain
conceptual complications (cf. \cite{NIP} or Theorem Nr. 2 in \cite{ali}).
One of the most serious ones is
that the mathematics of the evolution becomes practically
prohibitively complicated
whenever the operator $A_j(t)$ of interest
(which is, by assumption, non-Hermitian in
our manipulation space ${\cal F}$)
remains time-dependent even after a return to its
conventional self-adjoint representation
in Schr\"{o}dinger picture with trivial metric at $j=N$.

For this reason
we will postulate, in what follows, that at $j=N$
(i.e., in effect, in Schr\"{o}dinger picture)
the observable of interest becomes
time-independent,
$A_N(t)=A_N(0)=\mathfrak{a}$.
Under this assumption we may recollect that
 \be
 \Omega_N(t)\,A_{N-1}(t)=\mathfrak{a}\,\Omega_N(t)\,
 \label{relat}
 \ee
(cf. the last line of
Table \ref{enhybr}).
Thus, in an almost complete
parallel with our preceding recursive
considerations we may
differentiate relation (\ref{relat}).
The result
 \be
 {\rm i}\frac{d}{dt}\,A_{N-1}(t)=
 A_{N-1}(t)\,\Sigma_N(t)-\Sigma_N(t)\,
 A_{N-1}(t)\,.
 \ee
can be perceived as the desired $j=N-1$ operator-evolution equation
in which the role of the generator
(i.e., of a ``false Hamiltonian'' \cite{[382]})
is played by precisely the same
Coriolis force as defined by Eq.~(\ref{udoresekv}) above.

In the same spirit we may differentiate the whole set of the
identities listed in Table \ref{enhybr},
 \be
 \Omega_{N-k}(t)\,A_{N-k-1}(t)=A_{N-k}(t)\,\Omega_{N-k}(t)\,,
 \ \ \ \ \
 k = 1,2, \ldots, N-1\,.
 \ee
The description of the HIP dynamics becomes completed.

\begin{thm}
\label{thmtri}
In the $j-$th physical Hilbert space ${\cal R}_0^{[j]}$ with $j=N-k$,
the time-evolution of any operator $A_j(t)$ representing an observable
is controlled by
Heisenberg equation
 \be
 {\rm i}\frac{d}{dt}\,A_{N-k}(t)=
 A_{N-k}(t)\,\Sigma_{N-k+1}(t)-\Sigma_{N-k_1}(t)\,
 A_{N-k}(t)\,.
 \ee
The role of the generator of evolution
is played by the composite Coriolis force of Eq.~(\ref{ktvudoresekv}).
\end{thm}

\subsection{Evolution equations for density matrices.}

In the NIP review \cite{NIP}
we considered, in our present terminology, just the
specific and non-hybrid $j=N$ theory.
We
pointed out that besides the knowledge of
the Hamiltonian and metric,
the description of dynamics
requires
also the factorization of the
metric
$\Theta(t)=\Omega^\dagger(t)\Omega(t)$
yielding the definition (\ref{doresekv}) of
the Coriolis force $\Sigma(t)$.

From a purely technical point of view
it is necessary to distinguish, first of all,
between the stationary NSP theory
(in which $\Theta(t)=\Theta(0)$ is time-independent)
and the non-stationary NIP theory
(in which the Coriolis forces emerge).
In both of these cases
a key technical role is played, in our
manifestly unphysical representation
space
${\cal F}$, by the
respective two
formulae
 \be
 |\psi^{(NSP)}(t)\kkt = \Theta^{(NSP)}(0)\,|\psi(t)\kt\ \ \ \ {\rm and}\ \ \ \
 |\psi^{(NIP)}(t)\kkt = \Theta^{(NIP)}(t)\,|\psi(t)\kt\,
 \label{forfort}
 \ee
(cf. the first line of Table \ref{huhybr}).
Obviously, in both of these formulae
it would make very good sense to
replace the
difficult task of the construction
of the {\em operator\,}
(i.e., in finite dimensions, of an $N$ by $N$
array of matrix elements of the metric) by the
construction
of the mere {\em vector\,}
(i.e., in finite dimensions, of the $N$
components of
$|\psi(t)\kkt$).

The latter observation sounds like a paradox
but its clarification is easy
because at any $j$
the physical Hilbert space ${\cal R}^{[j]}_0(t)$
is always
represented in the same auxiliary mathematical space
${\cal F}$.
In other words, it is sufficient to
recall Table \ref{huhybr} and to
work simply
with the sets of the $j-$dependent
biorthonormalized bases in ${\cal F}$.

In such a biorthonormal basis setting
it is only necessary to remind the users that
even the pure
state of the quantum system in question is
represented by the {\em pair\,} of the ket vectors
or, more formally,
by the elementary projector
 \be
 \pi_{\psi}(t)=|\psi(t)\kt \,\frac{1}{\bbr
 \psi(t)|\psi(t)\kt}\,\bbr \psi(t)|\,
 \label{[26]}
 \ee
(cf. also a more extensive commentary in \cite{Brody}).
Incidentally, such an observation simplifies transition to a
statistical version of QHQM
where the states become represented by the
\textcolor{black}{non-Hermitian and time-dependent version
of the density
matrix (cf., e.g., \cite{[thi6]})},
 \be
 {\varrho}(t)=\sum_{k}|\psi^{(k)}(t)\kt
 \,\frac{p_k}{\bbr \psi^{(k)}(t)|\psi^{(k)}(t)\kt}\,\bbr
 \psi^{(k)}(t)|\,.
 \label{[27]}
 \ee
Here, the ``weights'' or ``probabilities''
of the $k-$superscripted individual pure-state components
have to be normalized to one,
$\sum_{k} p_k=1$.

Let us now turn attention to the consequences.
First of all let us, once more, choose the non-hybrid NIP option
with $j=N$, i.e., let us
assume that $\Theta=\Theta(t)$ (needed, first of all,
in the definition (\ref{forfort}))
can vary with time.
In such a case
the
Coriolis force $\Sigma(t)$
of Eq.~(\ref{doresekv})
starts playing, suddenly, a decisive
dynamical role.
By the differentiation of the operator of metric
$\Theta(t)=\Omega^\dagger(t)\Omega(t)$
one obtains, first of all, the law
 \be
 {\rm i}\,\frac{d}{dt}\,\Theta(t)=
 \Theta(t)\,\Sigma(t)-\Sigma^\dagger(t)\,
 \Theta(t)\,
 \ee
controlling its
evolution in time.
In the light of our previous comment, nevertheless,
it makes sense to circumvent the solution of
such an equation for operators
in favor of the incomparably more economical
solution of Schr\"{o}dinger
equation (\ref{sgvuasekv}), i.e.,
 \be
 {\rm i}\frac{d}{dt}\,|\psi(t)\kkt= G^\dagger(t)\,|\psi(t)\kkt\,,
 \ \ \ \
 G^\dagger(t)=H^\dagger(t)-\Sigma^\dagger(t)\,.
 \label{reresekv}
 \ee
The latter equation has the form of the adjoint
partner
of Eq.~(\ref{resekv}) in ${\cal F}$.
It is important to notice
that, as a consequence,
the overlaps $\bbr \psi(t)|\psi(t)\kt$ remain constant in time.

Once we
recall the definitions
of states
(\ref{[26]}) or (\ref{[27]}) (with constant
weights $p_k \neq p_k(t)$),
the explicit construction of the metric
(i.e., of the operator which is, due to relation (\ref{forfort}), implicitly
present in $|\psi(t)\kkt$)
may be declared redundant \cite{NIP}.
As the most important
byproduct of such a simplification we finally
deduced, in \cite{NIP},
the equation
 \be
 {\rm
 i\,}\partial_t\, {\varrho}(t)= G(t){\varrho}(t)
 -{\varrho}(t)\,G(t)\,
 \ee
which formally
controls the evolution of the non-hybrid density matrices in Liouvillean
picture.

On this background the transition to the hybrid formalism is immediate.

\begin{thm}
\label{thmctyri}
In the $j-$th physical Hilbert space ${\cal R}_0^{[j]}$
with $0 \leq j \leq N$
the evolution of the mixed state represented by the
density matrix
 \be
 {\varrho}_j(t)=\sum_{k}|_{[j]}\psi^{(k)}(t)\kt
 \,\frac{p_k^{[j]}}{\bbr_{[j]} \psi^{(k)}(t)|_{[j]}\psi^{(k)}(t)\kt}\,
 \bbr_{[j]} \psi^{(k)}(t)|\,
 \label{b[27]}
 \ee
is prescribed by the
evolution equation
 \be
 {\rm
 i\,}\partial_t\, {\varrho}_j(t)= G_j(t){\varrho}_j(t)
 -{\varrho}_j(t)\,G_j(t)\,.
 \ee
\end{thm}
\begin{proof}
It is sufficient to recall
the Schr\"{o}dinger equations
of the Theorems~\ref{thmjedna} and \ref{thmdve}
controlling the time-evolution of the respective
state-vector kets
$|_{[j]}\psi(t)\kt$ and
$|_{[j]}\psi(t)\kkt$. Subsequently,
one just has to recall the triviality of the conventional
Hermitian conjugation in the
unphysical representation space~${\cal F}$.
\end{proof}


\section{\textcolor{black}{An exactly solvable illustrative
example of applicability
of the FNIP formulation of quantum mechanics\label{appb}}}


\textcolor{black}{One of the best {\em non-stationary\,} illustrative examples
supporting and motivating the
applicability
of the present
factorization-based NIP technique
may be found in Ref.~(\cite{FT}).
In the methodical setting the latter model
parallels its stationary predecessor of Ref.~\cite{BG}.
Hence, its main merit
lies
in the fact that in the factorization ansatz (\ref{fffs})
the number $N$ of the factors
(which partially depends on the way you count them)
is, in any case, nontrivial, $N>2$.}

\textcolor{black}{From the point of view of
mathematics the latter model is in fact exceptional
because {\em both\,} the initial, dynamical input information
about the system {\em and\,} the ultimate,
Hamiltonian-operator-based information about the
probabilistic interpretation of the unitary quantum system
in question are carried by the {\em local potentials\,}
entering the respective linear, ordinary differential
operators.}

\textcolor{black}{
In the literature the task is most often found feasible just in
the stationary version of the formalism
(for details see, say, the recent review paper
\cite{ali}). In such a context we recommended,
in \cite{[381]}, just
an innovative approach to the
Hermitizations
in which
one has a choice
between $N$
alternative
stationary
representations of the system.}

\textcolor{black}{
In the more general non-stationary case
(see its introduction in \cite{timedep}) the
emerging larger set of
the necessary evolution equations often becomes
almost prohibitively complicated
(see also the discussion of this point in section \ref{sectionthree} above).
For this reason, any simplification of the task
(including the one based on the non-stationary extension of the
alternative-representation idea of Ref.~\cite{[381]}) may be expected to play
a key role in applications.
}

\textcolor{black}{ In paper \cite{FT} the search for an exactly
solvable illustrative example of the general non-stationary NIP
Schr\"{o}dinger equation (i.e., in our present notation, of the
ordinary differential equation of the form (\ref{resekv})) was based
on a suitable ansatz (\ref{fffs}) where the Dyson map has been
factorized using $N=4$. Its separate factors were chosen in the
elementary exponential-operator forms
 \be
 \Omega_4(t)=e^{\alpha (t)x}\,,\ \
 \Omega_3(t)=e^{\beta (t)p^{3}}\,,\ \
 \Omega_2(t)=e^{i\gamma (t)p^{2}}\,,\ \
 \Omega_1(t)=e^{i\delta(t)p}\,
 \label{etas}
 \ee
where $p$ denoted the momentum and where the optional coefficient
functions $\alpha(t)$, $\beta(t)$, $\gamma(t)$ and $\delta(t)$ of
time $t$ were merely assumed sufficiently smooth and real. }

\textcolor{black}{ In terms of these functions the Coriolis force of
Eq.~(\ref{resekv}) was then easily evaluated to acquire closed
form
 $$
 \Sigma(t)= ix\dot{\alpha}+i\dot{\beta}p^{3}-(3\dot{\alpha}
 \beta +\dot{\gamma})p^{2}-(2i\gamma \dot{\alpha}+\dot{\delta})p-i\delta \dot{
 \alpha}\,.
 $$
This being specified the authors of \cite{FT} decided to analyze the family of the
unitary quantum systems defined by a specific generator possessing
the following ``wrong-sign'' anharmonic-oscillator form,
 $$
 G(t)=G(z,t)=p^{2}+\frac{m(t)}{4}z^{2}-\frac{g(t)}{16}z^{4}\,.
 $$
Both the mass term $m(t)$ and the coupling $g(t)$ were assumed real
and smooth while the choice of a suitable complex contour of $z$ was
inspired by Jones and Mateo \cite{JM}, i.e., specified by formula
$z=-2i\sqrt{1+ix}$ where $x\in \mathbb{R}$.
This made the Schr\"{o}dinger equation defined on a
real line of the ``coordinate'' $x$ (entering also ansatz
(\ref{etas})) so that the straightforward reference to the
definition of the observable Hamiltonian $H(t)=G(t)+\Sigma(t)$
(denoted, in \cite{FT}, by a tilded symbol $\widetilde{H}(t)$) led
to an immediate evaluation of the isospectral partner Hamiltonian
$\mathfrak{h}(t)$ which has to be, by definition, self-adjoint in
${\cal L}$. }

\textcolor{black}{The latter condition (i.e., constraint
(\ref{hershaft})) appeared to specify all of the functions
$\alpha(t)$, $\beta(t)$, $\gamma(t)$ and $\delta(t)$ (which defined
the preselected and factorized Dyson map) in terms of a real
integration constant $c_1$ and of the two functions $m(t)$ and
$g(t)$ which characterized the preselected generator $G(t)$ (all of
these relations were made explicit by Eq. Nr. (2.11) in \cite{FT}). }

\textcolor{black}{The latter result became feasible due to a certain
auxiliary functional constraint as given by relation Nr. (2.12) in
\cite{FT}. As a consequence, a remarkable reparametrization of
$G(t)$ has been obtained in which $m(t)$ and $g(t)$ appeared
reparametrized in terms of another integration constant $c_2$ and of
an ``arbitrarily reparametrized time''
 $
 \sigma(t)
 $ and its derivatives,
 \be
 g(t)=\frac{1}{4\sigma^{3}(t)}\,,\ \ \ \ \  \
 m(t)=\frac{4c_{2}+\dot{\sigma}^{2}(t)
 -2\sigma(t) \ddot{\sigma}(t)}{4\sigma^{2}(t)}\,.
 \ee
In terms of the latter ``new time'' function $\sigma(t)$ is was
finally possible to write down an explicit (albeit discouragingly
lengthy) formula for the ``conventional textbook'' Hamiltonian
$\mathfrak{h}(t)$ in ${\cal L}$ (cf. formula Nr. (2.14) or, in the
special case with $c_1=c_2=0$, formula Nr. (2.15) in \cite{FT}). }

\textcolor{black}{In this manner the Hermitization $H(t) \to
\mathfrak{h}(t)$ has been completed. Nevertheless, the authors of
paper \cite{FT} also developed an alternative version of the
Hermitization in which one needs to use another ansatz (\ref{fffs})
with a larger $N>4$ and in which the result of the isospectral
mapping $H(t) \to \hat{\mathfrak{h}}(t)$ leads to a less exotic and fully
conventional form of the self-adjoint Hamiltonian $
\hat{\mathfrak{h}}(t)= p^2_y + \mathfrak{v}(y)$. The kinetic
energy term $p^2$ is well separated here from the local and confining
double well potential given, in {\it loc. cit.}, by formula Nr.
(2.19).}


\section{Discussion\label{sectionfive}}

\textcolor{black}{The strength of our present FNIP
approach to the quantum mechanics of closed and unitary
systems lies in
its potential phenomenological applicability
as illustrated by the
natural emergence of the factorization
of the Dyson map in the analytically solvable model of preceding section.
From an alternative, purely pragmatic and numerically oriented point of view
another persuasive motivation and support of the use of
the flexible FNIP
menu has been provided
by a schematic two-state model in Ref.~\cite{[382]} where
strictly one of the representations of the system in question
appeared most user-friendly and, hence, optimal and, from the point of view of
maximal economy, unique. Naturally, the formal and mathematical
preferences of such a type
will always be strongly physics-dependent.}

\subsection{\textcolor{black}{The occurrence of non-stationary non-Hermitian
observables in applications.}\label{udruhesekce}}

\textcolor{black}{
In the context of physics there exist many quantum systems of interest
in which the non-stationary and non-Hermitian versions of observables
play a key role. Besides the highly abstract cosmological problems
and besides the methodical challenges related to the
quantization of gravity (in which a key role is played by the
non-stationary and
non-Hermitian but Hermitizable
Wheeler-DeWitt
equation \cite{Thiemann}), the need of use of
the same, very general class of non-stationary operators also occurs
in the gauge field models \cite{xian},
or in the
atomic and molecular systems  \cite{[112],[113],[thi2]},
or in the field-theory-mimicking
spin systems exposed to a
time-dependent external field \cite{[26]}, etc.
}

\textcolor{black}{
One of the oldest
successful application of a realistic non-Hermitian Hamiltonian of a closed,
unitary quantum system may be found in the
Dyson's phenomenologically motivated
study \cite{Dyson}
devoted to the phenomenon of ferromagnetism.
The main idea lied in the decision of
treatment of a conventional bound-state wave function $|\psi(t)\pkt$
living in a conventional physical Hilbert space ${\cal L}$ as a result of action (\ref{anz})
of a suitable operator $\Omega(t)$ upon a (presumably, simpler) element  $|\psi(t)\kt$
of a (presumably, user-friendlier) auxiliary Hilbert space ${\cal F}$.
Naturally, the use of ansatz (\ref{anz}) can be perceived as a more or less
routine separation and sharing of the degrees of freedom hidden in $|\psi(t)\pkt$
between a ``simplified state vector''  $|\psi(t)\kt$ and an educated guess of some of the
most relevant ``correlations''  $\Omega(t)$.
}

\textcolor{black}{
In the latter sense the Dyson's study may be perceived
as paralleled by several analogous constructive techniques
involving, e.g., the so called
interacting boson model philosophy used in nuclear physics \cite{[7],[8],[9]},
or the flexible
coupled cluster method of
atomic and molecular physics (CCM, cf. Refs. \cite{[2],[6]}), etc.
}


\textcolor{black}{
Even in the Dyson's times the potential efficiency of
ansatz (\ref{anz}) was not unknown. Typically, it was used
in the traditional
Hartree-Fock approach to many-body systems where $|\psi(t)\kt$ is usually chosen
equal to an elementary determinant \cite{Messiah}). Nevertheless,
the essence of the
Dyson-inspired innovation lied in the manifest non-unitarity of
the mapping (cf. (\ref{nonuni})).
Naturally,
even with the stationary version of such an ansatz using constant operators
$\Omega(t)=\Omega(0)=\Omega$ and $\Theta(t)=\Theta(0)=\Theta$
the work with ansatz (\ref{anz}) appeared to lead to
the emergence of multiple serious technical
challenges and obstacles. Their concise list may be found,
e.g., on p. 1216 of review \cite{ali}. This was probably the main reason why even
the applications of the simpler, stationary NSP
version of the formalism
remained restricted, for a long time,
just to some specific physical systems, say,
in nuclear physics \cite{Jenssen}.
The situation only changed after
Bender with Boettcher \cite{BB}
turned attention of the community to the
subset of quasi-Hermitian models
exhibiting the parity-time symmetry
(${\cal PT}$-symmetry, cf. also Appendix A for more comments).
}

\subsection{\textcolor{black}{A competition between} alternative formulations of quantum mechanics}

In the conventional studies of unitary quantum
systems
one mostly has a choice between the theory in
Schr\"{o}dinger picture (SP)
and in
Heisenberg picture (HP).
Occasionally, one encounters also
the third methodical
option called intermediate picture (IP,
cf., e.g., p. 321 in \cite{Messiah}).
The criteria of the choice are usually pragmatic and
purpose-dependent. Thus,
the SP approach appears mathematically
most economical, especially when all of the observables
of interest remain
stationary. In contrast,
the HP ``representation'' is
attractive due to its closest
formal analogy with classical mechanics.
In comparison, the most general IP formulation of quantum mechanics
has traditionally been perceived as a mere technical
complement and
interpolation between the two fundamental
extremes.

In the methodically innovated and generalized QHQM
framework
of Scholtz et al \cite{Geyer},
all of the three (viz, SP, HP and IP) approaches
appeared to play equally fundamental roles
(cf.
\cite{Carl,ali,book}).
In the present paper we paid attention
to the specific
non-Hermitian IP (NIP) philosophy
emphasizing that
the main source of its relevance
lies in
its flexibility.
Indeed, the NIP formalism assumes the knowledge
of certain non-Hermitian
operators of observables with real spectra
(a.k.a. quasi-Hermitian operators \cite{Dieudonne})
which are, in principle at least, Hermitizable
(cf. \cite{SIGMA,timedep,NIP,Bila,Fring}).
This means that the NIP formalism is fully compatible
with the standard textbook postulates of quantum theory.

Such a Hermitization
``involves the construction of a
metric'' \cite{Geyer}.
In the older, NSP stationary approach
the physical Hilbert-space metric
$\Theta$
must necessarily (or at least {\it de facto},
up to an overall multiplication factor) remain  stationary
(cf., e.g., Theorem 2 in \cite{ali}). Currently, it is widely accepted
that
the same metric-stationarity requirement
may be also encountered
beyond the non-Hermitian SP setting, viz.,
in the non-Hermitian
HP approach \cite{NHeisenberg}.

The methodical relevance and practical
importance of the recent mathematical developments
is currently
encouraging
an increase of popularity of
non-Hermitian formalism
\textcolor{black}{beyond the strictly specified
domain of
physics of the closed and
hiddenly Hermitian (i.e., NSP or NIP)
quantum systems.
{\it Pars pro toto}, let us mention
the enormous productivity of the transfer
of the related innovative mathematics, say, to the
quantum theory of open systems \cite{Nimrod}
or to the even less standard nonlinear effective
theories and methods \cite{[thi1],[thi4]}.
}

\subsection{\textcolor{black}{Three most important Hilbert spaces of the theory}}

More than half a century ago,
during his numerical study of a
certain realistic many-body quantum
system Freeman Dyson \cite{Dyson}
revealed that
an optimalization of his
computations can be achieved via
a suitable modification
of the standard diagonalization techniques of textbooks.
The efficiency of
some of his calculations appeared
enhanced when
he replaced
the conventional (and time-independent) self-adjoint
Hamiltonian $\mathfrak{h}$
by its
isospectral
avatar $H$
of Eq.~(\ref{resekv})
where the operator
$\Omega\,$ was not required unitary.
Later on,
the judicious choice of $\Omega$
also succeeded
in a decisive simplification of some other
realistic calculations \cite{Jenssen}.


From a purely methodical perspective
it was not too difficult
to keep in mind that
the self-adjointness of $\mathfrak{h}$ in ${\cal L}$
was equivalent to the $\Theta-$quasi-Hermiticity
of $H$ in ${\cal F}$.
Such an observation
implies that
one might have a choice
between
the
direct
diagonalization of
$\mathfrak{h}$
and
the
preconditioning $\mathfrak{h} \to H$.
At the same time, such a freedom of decision
only emerges when the manifestly self-adjoint
Hamiltonian $\mathfrak{h}$
appears {\it much\,} more complicated than $H$.
One of the reasons is that in contrast to the
entirely elementary
Hermitian-conjugation rule in ${\cal L}$,
the use of its metric-dependent NSP alternative
requires certain additional and, for some
specific models
of quantum dynamics, truly
nontrivial
mathematical analysis \cite{Siegl,Viola,Iveta,Uwe}.

The introduction of a non-trivial QHQM metric $\Theta \neq I$
had to mediate, first of all, an (indirect) return
of the Hamiltonian $H$
to the desirable,
unitarity-compatible
ultimate mathematical status of a ``hidden'' Hermiticity.
Thus,
one could treat its quasi-Hermiticity as a constraint in which
$H$
plays the role of an input information about dynamics
(unconstrained by the over-restrictive Hermiticity requirement \cite{BB})
while the metric
is interpreted as one of the available
Hermitization tools \cite{Lotoreichik}.

The motivation of the split and
the main reason for the introduction of
another space has been found to
lie in the related possibility of a
sufficient
(i.e., in practice, drastic)
simplification of the
inner product. Thus,
in ${\cal F} \equiv {\cal H}_{unphysical}$ one can use the Dirac's
bra-ket convention and write
$\br \psi_1|\psi_2\kt_{(unphysical)}=\br \psi_1|\psi_2\kt$.
In parallel,
the ``correct physical'' inner product in ${\cal H}_{physical}$
can still be represented, in
the preferable and
decisively user-friendlier
auxiliary ``mathematical'' Hilbert space ${\cal H}_{unphysical}$, by
a closed formula
$\br \psi_1|\psi_2\kt_{(physical)}=\br \psi_1|\Theta|\psi_2\kt$.

The
costs of the simplification
of the Hilbert space
${\cal H}_{physical}\to {\cal H}_{unphysical}$
must be acceptable.
Among them,
the best visible item to check is that the (obligatory)
self-adjointness of the observables in ${\cal H}_{physical}$
has to be replaced by the respective
quasi-Hermiticity constraints \cite{Dieudonne}.
In this manner the QHQM formalism can be briefly
characterized by the triplet of the Hilbert spaces
according to the following diagram,
 \be
 \label{krexhu}
  \ba
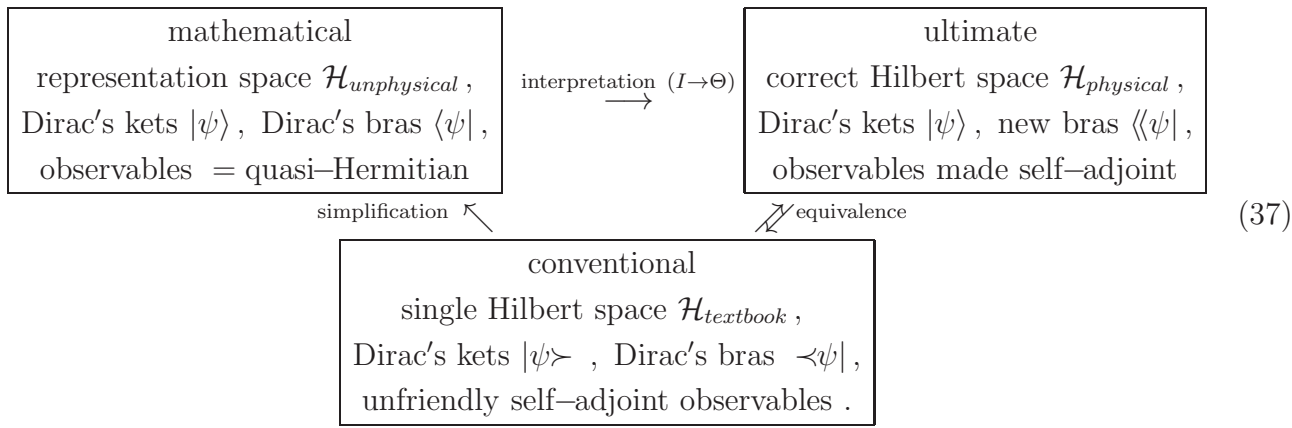
 \begin{array}{|c|}
 \hline \
   {\rm mathematical}\\
 {\rm representation\  space\ }
 {\cal H}_{unphysical}\,, \\
 {\rm Dirac's\  kets\ } |\psi\kt\,,\
 {\rm Dirac's\  bras\  } \br \psi |\,,
 \\
 {\rm observables\  } =
  {\rm quasi\!\!-\!\!Hermitian}
  \\
  \hline
 \ea
\
 \stackrel{{\rm interpretation\ }(I \to \Theta )}{ \longrightarrow }
 \begin{array}{|c|}
 \hline \
   {\rm ultimate}\\
 {\rm correct\  Hilbert\ space\ }
  {\cal H}_{physical}\,,
 \\
 {\rm Dirac's\ kets\ } |\psi\kt\,,\
 {\rm new\ bras\ } \bbr \psi |\,,
 \\
 {\rm observables\ }
 {\rm made\ self\!\!-\!\!adjoint}
 \\
  \hline
 \ea
 \\
 \ \ \
 \ \ \ \ \ \ \ \
  \stackrel{{\rm  simplification\, }\ }{ }   \nwarrow  \
 \ \ \ \ \ \
 \ \ \ \ \
 \ \ \ \ \ \ \ \
 \ \ \ \
  \  \nearrow \!\!\!\!\swarrow
  \stackrel{{\rm equivalence\ }}{ }
 \ \ \ \ \
\ \ \ \ \\
 \begin{array}{|c|}
  \hline \
   {\rm conventional}\\
   {\rm single\ Hilbert\  space\ }
 {\cal H}_{textbook}\,, \ \\
 {\rm Dirac's\  kets\ } |\psi\pkt\,,\
 {\rm Dirac's\  bras\  } \pbr \psi |\,,
 \\
 {\rm unfriendly\ self\!\!-\!\!adjoint\ observables\  }
 {\rm .}
 \\
  \hline
 \ea
 \\
\ea
 \ee
One can summarize that
any  `dynamical input'' knowledge of
a suitable set of operators characterizing
the observable aspects of a quantum system
in ${\cal H}_{unphysical}$
can lead to a consistent
theoretical Hermitization interpretation
in two ways, viz,
via a return to  ${\cal H}_{textbook}$, or
via a transition to ${\cal H}_{physical}$\,.

Perhaps, the former reconstruction
recipe
has rarely been found feasible. Still, its rather exceptional exemplification
has been offered by Buslaev and Grecchi \cite{BG}
\textcolor{black}{(see also Appendix~A below)}.
They were able to reconstruct a self-adjoint
textbook isospectral partner $\mathfrak{h}$
(acting in ${\cal H}_{textbook}=L^2(\mathbb{R}$)
of a preselected
non-Hermitian operator $H$
(introduced
and defined as acting in ${\cal H}_{unphysical}$)
using strictly analytic means.

\subsection{Hybrid versions of the formalism}

The currently increasing
popularity of the
name ``non-Hermitian''
is partially undeserved
because
its mathematically more correct
version would be ``quasi-Hermitian'' (cf. \cite{Dieudonne}).
The reason is that
one can still easily
restore the probabilistic physical-state
meaning of the kets
$|\psi(t)\kt$
by the mere {\it ad hoc\,} change
of the ``false but standard'' inner product
$\br \psi_a|\psi_b\kt$ in ${\cal F}$.
The point
is that
the
restoration (with its discovery
dating back to review \cite{Geyer})
is straightforward because it is entirely sufficient
to introduce another, third Hilbert space (say, ${\cal H}_{physical}$)
in which the ket-vector elements remain the same
(i.e., we can equally well write $|\psi\kt \in {\cal H}_{physical}$)
and we only have to amend the inner product properly.

In
review \cite{Geyer} we
read that in the quasi-Hermitian NSP framework
``the normal quantum-mechanical interpretation \ldots
involves the construction of a metric (if it exists)''.
In other words, for a return to
${\cal H}_{physical}$
it is sufficient to keep working in ${\cal F}$
(with its conventional Dirac's
bra-ket notation and the current, manifestly unphysical inner product
$\br \psi_a|\psi_b\kt=\br \psi_a|\psi_b\kt_{\cal F}$)
and just introduce its correct physical upgrade in
${\cal H}_{physical}$.

A new and promising compromise between
the use of the trivial and nontrivial physical Hilbert-space metric
has recently been found in
a new, ``hybrid'' reformulation
of the formalism of quantum mechanics in \cite{[380]}.
In it one combines
a {\it partial}, intermediate Hermiticity-violating
preconditioning
of the Hamiltonian $\mathfrak{h} \to H_{1}$
(cf. equation Nr.~13 in \cite{[380]})
with a {\it partial\,} simplification $\Theta \to \Theta_{2}$
of the Hermitian-conjugation-mediating operator {\it alias\,} Hilbert-space
inner-product metric
(cf. equation Nr.~14 in \cite{[380]}).
An
explicit, fairly persuasive demonstration of the
(possible) optimality of the compromise
has been also provided
via an illustrative two-state quantum model
(cf. equations Nr.~24 and 25 in \cite{[380]}).

In a more mathematical spirit let us add that
in our preparatory studies~\cite{[381]} and \cite{[373]} a
fairly complicated notation
was needed to explain the idea of an iterative Hermitization
at $N > 0$.
We had to introduce sophisticated superscripts
marking the individual non-equivalent Hermitian conjugations,
e.g., in the auxiliary sequence of the
Hilbert spaces ${\cal R}_k^{[0]}$ numbered by the subscript.
In the present paper we managed to
simplify the discussion by assuming that all of the
relevant mathematical
entities (i.e., the vectors and the operators)
are defined as represented in the single,
manifestly unphysical but still
mathematically optimal working Hilbert space
${\cal R}_N^{[0]}$ {\it alias\,}
${\cal F}$ with the conventional trivial inner product.

Such a change of notation allowed us to summarize the
essence of the hybrid NSP formalism
in Table \ref{huhybr}. After our transition to the non-stationary NIP
framework it was then sufficient to add
there the indication of the
time-dependence of the maps $\Omega_j=\Omega_j(t)$ and of the metrics
$\Theta^{[j]}=\Theta^{[j]}(t)$.
In comparison with the above-mentioned older studies
we also postulated there that at every $j$ the kets and bras form a bicomplete
and biorthonormal basis in ${\cal F}$.
Thus, even the more mathematically rigorous
readers may still accept Table~\ref{huhybr}
as corresponding to the quantum models for which
the physical Hilbert space in question is finite-dimensional.
Whenever one needs to work with $\dim {\cal R}_0^{[j]} = \infty$,
a complementary reading may be sought, say, in the
dedicated reviews in \cite{book}.

In all three Tables \ref{hybr}, \ref{enhybr} and \ref{huhybr}
one could replace the word ``Hamiltonian''
by the name of any other observable. Nevertheless, our knowledge
of the Hamiltonian $H(t)$
plays a key role in the theory
because it enables us to recall the Stone theorem
\cite{Stone}
and
to relate
the Hermiticity of the Hamiltonian $\mathfrak{h}(t)$
guaranteeing
the unitarity of the evolution
 to the
mathematically equivalent
condition of the
quasi-Hermiticity of the Hamiltonian $H(t)$
in ${\cal F}$.

\section{Summary\label{sectionsix}}

The Dyson's very old idea \cite{Dyson}
of the description of a unitary quantum
system using a manifestly
non-Hermitian
version of its Hamiltonian $H=H^\dagger$
was treated, for a long time, as a special
mathematical trick rather than a deeper
innovation of the theory.
The change only came when
some of the basic conceptual problems
were shown tractable using several additional
mathematical assumptions:
Among them, the most important ones were
the
restrictions of the class of the admissible
observables to the operators which have to be, e.g., bounded
(cf. Scholtz et al
\cite{Geyer})
or pseudo-Hermitian {\it alias\,}
${\cal PT}-$symmetric
(see \cite{BB}
and reviews \cite{Carl,ali}).

After these developments it became widely accepted that
the resulting QHQM theory
offers a fully consistent
theoretical framework and
picture of the evolution of unitary
quantum systems.
In practice, the broadening of the scope of the theory
was accompanied by a steady development of its
technical aspects which made the formalism applicable
not only to the models using the stationary inner products
(in the NSP framework)
but also in the non-stationary NIP scenario.

In our present paper we paid attention to the latter
methodical context and we proposed its efficient
simplification based on the
general $N-$term factorization of the
underlying Dyson transformation operators $\Omega(t)$.
We may summarize that our present main result
may be seen in the explicit description of the consequences
of the latter factorization ansatz. We
revealed that in contrast to the initial scepticism
as expressed in our preparatory study \cite{[382]},
the resulting new FNIP formulation of quantum mechanics
still remains formally transparent and user-friendly.

The core of the survival of the simplicity of
our formalism may be seen in the discovery of the
comparatively elementary formulae which define the
operators (which we called quantum Coriolis forces)
in recurrent way. Precisely this enabled us to
show that the complexity of the
set of the underlying evolution equations
(in which the Coriolis force forms the essential
ingredient)
remains acceptable and that it
basically does not grow with $N$.
Moreover, from the theoretical perspective we may say that
the main new feature of the
formalism is that it enables its user to make a choice
of the correct physical Hilbert space ${\cal R}_0^{[j]}(t)$
out of the menu of $N+1$ eligible candidates
such that $0\leq j \leq N$
(here we included also $j=0$ which would only mark the
conventional Hermitian theory).

Our new formalism incorporates several
existing versions of
the non-stationary non-Hermitian quantum mechanics.
Besides the trivial Hermitian theory with $j=0$
we have to mention
the NIP theory
where
there is no hybrid option
because $N=1$.
Similarly, just a single hybrid option
has been found tractable in
\cite{[382]}
so that no variable index $j$ has been needed, either.
In this sense, the variability of $j-$options
in
our present
quantum theory of bound states
can be interpreted as a new degree of
flexibility, i.e., a new method
of the search for an optimality of
the practical quantum model-building.


\newpage

\section*{\textcolor{black}{Appendix A. Exactly solvable models
using factorized NSP Dyson maps $\Omega\neq \Omega(t)$}\label{appa}}

Our present project has been motivated, first of all,
by the Buslaev's and Grecchi's \cite{BG} discovery that
one of the new and promising
ways towards simplifications could be sought
in the factorization ansatz (\ref{fffs}).
More systematically, the idea
of the decomposition of the
Dyson
map into its more elementary operator components
has been developed in our recent
studies \cite{[380],[381]}.
We demonstrated there, constructively, that
the use of the
composite forms of the preconditioning operator
may really offer,
in the stationary NSP framework at least,
a perceivable simplification of at least some
of the necessary calculations.

\textcolor{black}{The emergence of an ambitious theoretical project of the representation
of an operator of a quantum observable (i.e., typically, of the Hamiltonian with real spectrum) by a
manifestly non-Hermitian operator is usually being traced back to the Dyson's paper \cite{Dyson}
on manybody quantum systems.
Nevertheless, some of the sources of interest in non-Hermitian
Hamiltonians $H \neq H^\dagger$ with real spectra are much older.
One of them dates back to the
quantum-field-theoretically motivated
studies of the behavior of the spectrum of the
elementary ordinary differential anharmonic oscillator
 \be
 H(g)=-\frac{d^2}{dx^2} + x^2 + g\,x^4\,,
 \ \ \ \ \ x \in \mathbb{R}\,.
 \label{dropp}
 \ee
Its older
perturbation-expansion analyses  \cite{Seznec,[6a]}
led, indeed, to certain surprising and more or less
purely empirical observations
that under certain conditions
the spectrum may stay real even when the coupling itself becomes negative,
making
the potential
asymptotically repulsive.
}

\textcolor{black}{
The subsequent growth of interest in the latter anomaly
(with a final resolution
provided by Buslaev and Grecchi \cite{BG}) climaxed in the
Bender's and Boettcher's
discovery \cite{BB} that
the loss of the Hermiticity of the Hamiltonians
need not necessarily be accompanied by
the loss of the reality of the spectrum.
In fact, for the
manifestly non-Hermitian linear differential operators
 \be
 H^{(BB)}(\lambda)=-\frac{d^2}{dx^2} +V^{(BB)}(\lambda,x)\,,
 \ \ \ \ V^{(BB)}(\lambda,x)=({\rm i}x)^\lambda\,x^2
 \label{bobot}
 \ee
an opposite observation has been made:
Their spectrum proved real in the whole open interval of
the exponents $\lambda \in (0,2)$ \cite{Carl}.
}

\textcolor{black}{
The later discovery proved influential and led to the
development of multiple new classes of phenomenological quantum models
(say, in relativistic quantum mechanics)
in which one was allowed to replace the Hermiticity of the Hamiltonian
by a suitable weaker constraint (see, e.g., the comprehensive reviews of these
developments in \cite{Carlbook}).
}

\textcolor{black}{
The consistent probabilistic interpretation of models (\ref{bobot})
appeared to be nontrivial \cite{book}.
Even when the mass term has been dropped in (\ref{dropp}),
Jones and Mateo \cite{JM} managed to prove that
in the limit $\lambda \to 2$
the problem becomes tractable non-numerically. These authors showed that
the two Hamiltonians
 \be
 H^{(JM)}=p^{2}-gx^{4}\,, \ \ \ \
 \mathfrak{h}^{(JM)}=\frac{p^{4}}{64g}-\frac{1}{2}
 p+16gx^{2},  \label{JonesH}
 \ee
are spectrally equivalent, i.e., that the
bound-state spectrum of the
purely
quartic ``wrong-sign'' anharmonic oscillator
(living in a manifestly unphysical
but friendly Hilbert space ${\cal F}$)
coincides with the spectrum of the
conventional double-well model $\mathfrak{h}^{(JM)}$
living in a conventional Hilbert space ${\cal L}$.}

\textcolor{black}{The technique of the proof differed from the one used in \cite{BG}.
Indeed, in a preparatory step, Jones with Mateo deformed the axis of coordinates
into a curved contour in the complex plane, $x\rightarrow -2i\sqrt{1+ix}$.
This simplified the analysis of the solutions in the asymptotic region.
Only then, a sequence of the isospectral transformations based on the factorization
(\ref{fffs}) has been used in the same manner as in \cite{BG}. In other words,
a sequence of the Hamiltonians
has been obtained in closed form, serving well as an illustration of the
factorization-based NSP stationary technique of Ref.~\cite{[381]}.
}

\end{document}